\documentclass[conference,letterpaper]{IEEEtran}

\usepackage{amsmath,amsthm,amsfonts,braket,algorithm,graphicx}

\usepackage[utf8]{inputenc} 
\usepackage[T1]{fontenc}
\usepackage{url}
\usepackage{ifthen}
\usepackage{cite}

\theoremstyle{definition} 
\theoremstyle{definition} 
\newtheorem {theorem} {Theorem}

\newcommand{\kb}[1]{\mathbf{[#1]}}

\newcommand{\bk}[1]{\braket{#1|#1}}

\newcommand{\MR}{\texttt{Measure and Resend}}
\newcommand{\R}{\texttt{Reflect}}

\newcommand{\ba}{\bar{a}}

\newcommand{\prf}{p^{A\rightarrow B}}
\newcommand{\prr}{p^{A\rightarrow A}}

\newcommand{\pe}{p^{AE}}

\floatname{algorithm}{Protocol}

\begin{document}
\title{From Classical to Semi-Quantum Secure Communication}
\author{
\IEEEauthorblockN{Allison Gagliano}
\IEEEauthorblockA{Departments of Mathematics \& Computer Science\\
Eastern Connecticut State University\\
Willimantic, CT 06226}
\and
\IEEEauthorblockN{Walter O. Krawec and Hasan Iqbal}
\IEEEauthorblockA{Computer Science \& Engineering Department\\
University of Connecticut\\
Storrs, CT 06268\\
Email: walter.krawec@uconn.edu}
}

\maketitle
\begin{abstract}
In this work we introduce a novel QKD protocol capable of smoothly transitioning, via user-tuneable parameter, from classical to semi-quantum in order to help understand the effect of quantum communication resources on secure key distribution.  We perform an information theoretic security analysis of this protocol to determine what level of ``quantumness'' is sufficient to achieve security, and we discover some rather interesting properties of this protocol along the way.
\end{abstract}

\section{Introduction}
A \emph{semi-quantum key distribution} (SQKD) protocol's goal is similar to that of a \emph{quantum key distribution} (QKD) protocol, namely the establishment of a secret key between two parties, Alice ($A$) and Bob ($B$), secure against an all-powerful adversary Eve ($E$).  Semi-quantum cryptography, first introduced in 2007 by Boyer et al., in \cite{SQKD-first-PRL}, imposes the restriction, however, that one of the users (typically $B$), is limited to being ``classical'' or ``semi-quantum.''  This restriction implies $B$ is limited to working only in the computational $Z$ basis (spanned by states $\ket{0}$ and $\ket{1}$).  He may not measure or prepare states in any other basis (we will discuss the exact capabilities of $B$ later in this paper).

The primary interest of these protocols is to help answer the question ``how quantum must a protocol be to gain an advantage over a classical one?'' We know that, if both parties are classical, key distribution is impossible unless computational assumptions are made. Thus, the question semi-quantum protocols seek to help answer is: what quantum resources are required to attain unconditional security?  However, besides removing certain key quantum capabilities from the two users, there has not been a semi-quantum protocol that can smoothly transition from classical to quantum allowing us to study the effects of quantum communication on secure key distribution.

In this paper, we propose such a protocol and analyze its properties. We introduce a novel SQKD protocol with a user-tuneable parameter $\alpha$ allowing one to, in a way, set the level of ``quantumness'' of the entire protocol.  Indeed, when $\alpha = 0$, the protocol collapses to a classical one (which is insecure).  As $\alpha$ increases, the protocol, in a way, becomes more quantum (in that Alice, the quantum user, is allowed to send and receive states which are less orthogonal).  However, Bob's capabilities, being classical in nature, are not affected by this $\alpha$ parameter. In fact, as the protocol becomes ``more quantum'' Bob has more trouble determining $A$'s key bit since $B$ is always restricted to the computational $\{\ket{0}, \ket{1}\}$ basis.

Our protocol is purely of theoretical interest.  We are interested in devising a way to measure the effect of quantum state generation and measurement on the security properties of a key-distribution system where one user is forced to be classical and as the other user varies in quantum capabilities.  We perform an information theoretic security analysis of our protocol and look at how $\alpha$ affects the noise tolerance of the protocol (i.e., how does the secure communication rate change as $A$ becomes more or less quantum, even when an all-powerful adversary is attacking).  Naturally, when $\alpha$ is too small, the protocol is ``too classical'' to be secure - as $\alpha$ increases the protocol can attain security for some noise levels; however once $\alpha$ increases too much, then Alice is ``too quantum'' for Bob to understand completely (i.e., he is unable to correctly guess what key-bit $A$ is trying to send to him).

We make several contributions in this work.  We introduce a novel SQKD protocol which is interesting theoretically as it is the first such protocol, that we are aware of, to allow researchers to gauge the effect of quantum state preparation and measurement on a key-distribution protocol where one user remains classical in nature.  This protocol is also highly restrictive in nature as $A$ and $B$ both have severe restrictions placed on them, yet we are still able to prove security.  Second, we perform an information theoretic security analysis of this protocol and our proof technique (which extends that of \cite{QKD-Tom-Krawec-Arbitrary} but to the highly restricted case where fewer noise statistics may be observed) may be of independent interest and applicable to other (S)QKD protocols where users are severely limited in their ability to measure the noise in the quantum channel (note that SQKD protocols require two-way quantum channels allowing Eve two opportunities to attack each qubit - this, in addition to the fact that $A$ and $B$ cannot observe all noise statistics due to their restrictions, greatly increases the complexity of the security analysis).  Finally, we evaluate our protocol, examining the effect of the $\alpha$ parameter for various channels and noise scenarios, discovering interesting properties along the way.

\subsection{Notation and (S)QKD Security}

We denote by $Z$ to be the \emph{computational basis} consisting of states $\{\ket{0}, \ket{1}\}$.  We use $H(p_1,\cdots, p_n)$ to be the Shannon entropy of $\{p_i\}$ and $H(x)$ to mean the binary Shannon entropy, namely $H(x) = H(x,1-x)$.  Note that all logarithms in this paper are base two.

Given a \emph{density operator} $\rho$ (that is, a Hermitian positive semi-definite operator of unit trace), we write $S(\rho)$ to be the von Neumann entropy of $\rho$ defined as $S(\rho) = -tr(\rho\log\rho)$.  If $\rho$ acts on Hilbert space $\mathcal{H}_A\otimes\mathcal{H}_B$, we often write $\rho_{AB}$.  In this case, we define $\rho_B$ to be the partial trace over the $A$ system, namely $\rho_B = tr_A\rho_{AB}$.  This notation extends to three or more systems.  To simplify notation, given $\ket{v}$ in some Hilbert space, we will write $\kb{v}$ to mean $\ket{v}\bra{v}$.

If $\rho_{AB}$ acts on $\mathcal{H}_A\otimes\mathcal{H}_B$, then we write $S(AB)_\rho$ to mean $S(\rho_{AB})$.  We also write $S(A|B)_\rho$ to mean the conditional von Neumann entropy defined to be $S(A|B)_\rho = S(\rho_{AB}) - S(\rho_B)$.  We will forgo writing the subscript ``$\rho$'' if the context is clear.

Any (S)QKD protocol requires both a quantum channel and an authenticated classical channel and these protocols operate in two stages.  The first, called the \emph{quantum communication stage}, utilizes the quantum channel and authenticated classical channel, over numerous \emph{iterations}, to agree on a so-called \emph{raw-key} of size $n$-bits.  Eve, who was attacking the quantum channel, and listening to the authenticated classical communication, also has an ancilla partially entangled with $A$ and $B$'s raw key.  At this point, the system ($A$ and $B$'s raw key along with $E$'s ancilla) may be represented by a classical-quantum state:
\begin{equation}\label{eq:cq-state}
\rho_{ABE} = \sum_{a,b \in \{0,1\}^n}P(a,b)\kb{a}_A\otimes\kb{b}_B\otimes\rho_{E}^{(a,b)}.
\end{equation}
From this, $A$ and $B$ run an error correction protocol (leaking additional information to $E$) and a privacy amplification protocol, shrinking the $n$-bit raw key to a secret key of size $\ell(n)$ on which $E$ has negligible information (in an information theoretic sense).  In the asymptotic scenario as $n \rightarrow \infty$, which we consider here, $E$'s information, and also all failure probabilities, go to zero.  An important statistic in any security proof is the \emph{key-rate:} $r = \ell(n)/n$.  For more information on these general concepts and definitions, the reader is referred to \cite{QKD-survey}.

As with almost all (S)QKD security proofs, we consider \emph{collective attacks}, whereby $E$ attacks the channel in an i.i.d. manner but is free to postpone her measurement of her ancilla to any future point in time and, indeed, may later perform an optimal \emph{coherent} measurement of her entire ancilla.  Usually, proving security against collective attacks is sufficient to prove security against general, arbitrary, attacks \cite{QKD-renner-keyrate,QKD-general-attack,QKD-general-attack2}.  We suspect this result also holds true for our protocol; however due to the highly restrictive nature of $A$ and $B$'s operation, a complete proof of this is outside the scope of this paper and would make for interesting future work.

Under a collective attack (in which case $\rho_{ABE}$ from Equation \ref{eq:cq-state} may actually be written $\rho_{ABE} = \sigma_{ABE}^{\otimes n}$ for some classical-quantum state $\sigma$), we may employ the Devetak-Winter key-rate equation \cite{QKD-winter-keyrate} which states:
\[
r = \lim_{n\rightarrow \infty}\frac{\ell(n)}{n} = \inf[S(A|E)_\sigma - H(A|B)],
\]
where the infimum is over all collective attacks which induce the observed statistics (e.g., the observed error rate, though one may also look at other statistics such as mismatched events \cite{QKD-Tom-First,QKD-Tom-KeyRateIncrease,QKD-Tom-Krawec-Arbitrary}).  It is this computation of $r$ (in particular, the computation of $S(A|E)$ since the computation of $H(A|B)$ is generally trivial) that is the key element in any (S)QKD security proof and our main focus in this work.  From this, one may look at a protocol's \emph{noise tolerance} - that is for what noise levels does $r$ remain positive.

In our security proof, we will make use of the following result proven in prior work (though slightly generalized here):
\begin{theorem}\label{thm:entropy}
(From \cite{QKD-Tom-Krawec-Arbitrary}): Given the classical-quantum state:
\[
\rho_{AE} = \frac{1}{N}\kb{0}_A\otimes\left(\sum_{i=1}^M\kb{E_i}\right) + \frac{1}{N}\kb{1}_A\otimes\left(\sum_{i=1}^M\kb{F_i}\right),
\]
then:
\begin{align*}
S(A|E)_\rho &\ge \sum_{i\in J}\frac{\bk{E_i}+\bk{F_i}}{N}\\
&\times \left(H\left[\frac{\bk{E_i}}{\bk{E_i}+\bk{F_i}}\right] - H\left[\lambda(\ket{E_i}, \ket{F_i})\right]\right),
\end{align*}
where:
\begin{equation}\label{eq:lambda}
\lambda(\ket{x}, \ket{y}) = \frac{1}{2}\left(1+\frac{\sqrt{(\bk{x}-\bk{y})^2+4Re^2\braket{x|y}}}{\bk{x}+\bk{y}}\right),
\end{equation}
and $J$ is any subset $J \subset \{1, \cdots, M\}$.
\end{theorem}
\begin{proof}
The proof for $J = \{1,\cdots, M\}$ can be found in \cite{QKD-Tom-Krawec-Arbitrary}.  The result also follows for arbitrary subset $J$ by noting that, in the proof, the term:
\[
H\left[\frac{\bk{E_i}}{\bk{E_i}+\bk{F_i}}\right] - H\left[\lambda(\ket{E_i}, \ket{F_i})\right]
\]
is the result of computing the conditional entropy of a classical-quantum state which is known to be always non-negative.
\end{proof}

\subsection{Semi-Quantum Cryptography and Related Work}
Since the framework's introduction in 2007 by Boyer et al., \cite{SQKD-first-PRL,SQKD-second}, numerous SQKD protocols have been proposed \cite{SQKD-first-PRL,SQKD-second,SQKD-lessthan4,SQKD-Single-Security,SQKD-3,SQKD-cl-A,SQKD-no-measure,SQKD-mirror} (just to list a few), some with information theoretic proofs of security \cite{SQKD-Krawec-SecurityProof,SQKD-zhang2016single,QKD-Tom-Krawec-Arbitrary}.  Often one is interested in removing requirements on one or both users while still attempting to attain security against an all-powerful adversary - this is to study the effects of these resources and abilities on the secure communication rate of the resulting protocol.  However, no prior SQKD protocols allow for the smooth transition from a purely classical protocol to a semi-quantum one.


An SQKD protocol requires a two-way quantum channel, allowing a qubit to travel from $A$ to $B$ (the \emph{forward direction}) and return from $B$ to $A$ (the \emph{reverse direction}).  $A$, the fully quantum user, is allowed to prepare any arbitrary quantum state and send it to the ``classical'' user $B$, who is allowed only to directly work with the $Z$ basis.  In more detail, on receiving a qubit, $B$ may choose to do one of two operations:
\begin{enumerate}
\item $\MR$: If he chooses this option, he performs a $Z$ basis measurement on the qubit, resulting in outcome $\ket{r}$, for $r \in \{0,1\}$.  He then resends the same state $\ket{r}$ to $A$.  Note that he can only measure and prepare qubits in this single basis.
\item $\R$: In this case, $B$ disconnects from the quantum channel and reflects all qubits back to $A$.  If this is chosen, $A$ is, essentially, communicating with herself.
\end{enumerate}
When a qubit returns to $A$, she is allowed to perform any quantum operation on it.  Note that, under this scenario, Eve is allowed two opportunities to attack every qubit.

\section{The Protocol}

Our protocol, being a semi-quantum one, requires a two-way quantum channel and forces $B$ to be ``classical'' in nature as described in the previous section.  We also place additional restrictions on the quantum user $A$.  On each iteration of the quantum communication stage, $A$ is allowed to send only one of two possible states: either $\ket{0}$ or $\ket{a} = \alpha\ket{0} + \beta\ket{1}$, where $\alpha\ge 0$ is a public, user-specified, parameter and $\beta = \sqrt{1-\alpha^2}$.

Bob is the classical user - as such, on receipt of a qubit from $A$, he may only directly interact with it through the $Z$ basis (by choosing $\MR$), or he may simply ignore the qubit and reflect it back to $A$ (by choosing $\R$).

When a qubit returns to $A$, she will perform a measurement using the three-outcome POVM $\Lambda = \{\Lambda_0, \Lambda_a, \Lambda_?\}$ defined:
$\Lambda_0 = p\ket{0}\bra{0},$ $\Lambda_a = p\ket{a}\bra{a},$
and where $\Lambda_? = I-\Lambda_0-\Lambda_1$.  The parameter $p$, which is another public constant, must be chosen to ensure $\Lambda_? \ge 0$.  Furthermore, $A$ wishes to maximize $p$ so that the probability she receives the indeterminate outcome ``?'' is minimized.  Some algebra reveals that the maximal $p$ that satisfies this is $p \le \frac{1}{1+\alpha}$.  Note that, in this work, where we only consider the asymptotic scenario, the actual choice of $p$ is not that important so long as $0 < p \le \frac{1}{1+\alpha}$.  In a finite key analysis, this choice of $p$ would be much more important, but we leave this as future work.

Notice that, when $\alpha = 0$, the protocol ``collapses'' to a purely classical communication system where $A$ sends $\ket{0}$ and $\ket{1}$ only and where she is always measuring in the $Z$ basis (since $p$ approaches $1$ as $\alpha$ decreases and so $\Lambda_0 = \ket{0}\bra{0}$, $\Lambda_a = \ket{1}\bra{1},$ and $\Lambda_? \equiv 0$).  Of course, $B$ is classical regardless of the choice of $\alpha$ since he is only able to measure and send in the $Z$ basis (or disconnect from the quantum channel, thus causing $A$ to simply ``talk to herself'').  For $\alpha > 0$, the protocol is inherently quantum - but the question is, how far from classical ($\alpha = 0$) must the communication be before we start attaining secure communication? Our protocol in detail is described in Protocol \ref{prot}.

\begin{algorithm}
\caption{$\alpha$-SQKD}\label{prot}
\textbf{Public, User-Defined, Parameters:} $\alpha$, the level of ``quantumness'' of the protocol; $p \in (0,1/(1+\alpha)]$, the POVM parameter as discussed in the text; $q \in (0,1)$, the probability that $B$ chooses $\MR$ on any particular iteration (in the asymptotic scenario, which we consider in this work, this value may be set arbitrarily close to $1$ as is done for other (S)QKD protocols to improve efficiency \cite{QKD-BB84-Modification,SQKD-Krawec-SecurityProof}).

$ $\newline
\textbf{Quantum Communication Stage:} The quantum communication stage repeats the following process using the two-way quantum channel and the authenticated classical channel:

$ $\newline
1. $A$ chooses a bit $k_A$ uniformly at random.  If $k_A = 0$, she sends $\ket{0}$ to $B$; otherwise she sends $\ket{a} = \alpha\ket{0} + \beta\ket{1}$, where $\beta = \sqrt{1-\alpha^2}$.
$ $\newline
2. $B$ chooses randomly to $\MR$ (with probability $q$) or to $\R$ (with probability $1-q$).  If he chooses $\MR$, he will save his measurement result as $k_B \in \{0,1\}$.
$ $\newline
3. Finally, with probability $q$, $A$ will simply discard the qubit; otherwise, she will measure using POVM $\Lambda$, as discussed in the text, saving the outcome (which is one of ``$0$,'' ``$a$,'' or ``$?$'').
$ $\newline
4. Using the authenticated classical channel, $B$ will disclose his choice of operation (either $\MR$ or $\R$) and $A$ will disclose whether she chose to measure or not.  For all iterations where $A$ chose to measure the returning qubit, $A$ will send to $B$ her preparation and measurement outcomes (\emph{these iterations will be used only to test the quantum channel and not for key distillation}).  For all other iterations (where $A$ did not measure) and if $B$ chose $\MR$, then $A$ and $B$ will use their respective $k_A$ and $k_B$ values to contribute towards their raw key.

$ $\newline
\textbf{Classical Reconciliation Stage:}
Following the quantum communication stage, assuming the channel noise is low enough (to be discussed), $A$ and $B$ will run error correction and privacy amplification, resulting in a secret key.
\end{algorithm}


The reader will observe that, for $\alpha > 0$, our protocol always has some noise in the raw key, \emph{even when no adversary is present}!  Indeed, unless the protocol is purely classical ($\alpha = 0$), the classical user $B$ will be unable to determine exactly the information that $A$ is trying to send.  The issue is exacerbated when an adversary comes into play (adding additional noise).  As mentioned in the introduction, \emph{the protocol is purely a theoretical one} studied for its theoretical interest to help study the ``gap'' between classical and quantum communication.  We do not expect this protocol to ever be implemented in practice (\emph{unless some faulty hardware forces this protocol to be used}).  Note that we are also not concerned with practical attacks such as photon loss or multi-photon states \cite{QKD-survey,SQKD-photon-tag,SQKD-photon-tag-comment} - though interesting, these issues are outside the scope of this theoretical analysis.

We are interested in two questions: Given an observed noise level $Q$, for what $\alpha$ is the protocol secure?  Of course when $\alpha = 0$, the protocol will never be secure.  Secondly, what is an optimal choice of $\alpha$?  That is, how ``far'' from the classical case of $\alpha = 0$ must the communication be to optimize the secure transfer of information between $A$ and $B$ when faced with a quantum adversary $E$.


\section{Security Analysis}
Our goal in this section is to compute our protocol's key-rate (specifically $S(A|E)$) as a function of $\alpha$ and those observable parameters that $A$ and $B$ may measure in the channel (which are very few).  We begin by deriving a density operator description of a single ``successful'' iteration of the protocol (where by ``successful'' we mean an iteration leading to the distillation of a raw key bit).  For now we assume collective attacks whereby Eve attacks each iteration in an i.i.d. manner.  In this case, as shown in \cite{SQKD-entangle}, for SQKD protocols, it suffices to only prove security against so-called \emph{restricted} collective attacks.  These restricted attacks consist of an isometry $\mathcal{F}: \mathcal{H}_T \rightarrow \mathcal{H}_T\otimes\mathcal{H}_E$ applied in the forward channel (connecting $A$ to $B$) and a unitary operator $U_R$ applied in the reverse channel and acting on $\mathcal{H}_T\otimes\mathcal{H}_E$.  Here we use $\mathcal{H}_T$ to denote the two-dimensional space modeling the qubit in transit and $\mathcal{H}_E$ is Eve's ancilla.  The action of $\mathcal{F}$ is simply:
\begin{align}
\mathcal{F}\ket{0}_T &= q_0\ket{0,0}_{TE} + q_1\ket{1,e}_{TE}\label{eq:forward-attack}\\
\mathcal{F}\ket{1}_T &= q_2\ket{0,f}_{TE} + q_3\ket{1,0}_{TE},\notag
\end{align}
where $q_i \in \mathbb{R}_{\ge 0}$ subject to $q_0^2 + q_1^2 = q_2^2 + q_3^2 = 1$ and where $\ket{e}$ and $\ket{f}$ are arbitrary, normalized, vectors in $\mathcal{H}_E$.  There are some additional restrictions that may be made on this attack (in particular $\ket{e}$ and $\ket{f}$ may exist with a two-dimensional subspace of $\mathcal{H}_E$ spanned by $\ket{0}_E$ and a second basis vector); however, this notation is sufficient for the discussion at hand.  For further information on the restricted attack, and the proof that security against such attacks implies security against arbitrary collective attacks, the reader is referred to \cite{SQKD-entangle}.  Note that, by linearity of $\mathcal{F}$, we also have the following:
\begin{align}
\mathcal{F}\ket{a} &= \ket{0}_T\otimes(q_0\alpha\ket{0}_E + q_2\beta\ket{f}_E)\label{eq:forward-attack-a}\\
&+ \ket{1}_T\otimes(q_1\alpha\ket{e}_E + q_3\beta\ket{0}_E).\notag
\end{align}

To build the desired density operator, we trace the evolution of an iteration of the protocol.  Following $A$'s preparation (randomly sending $\ket{0}$ or $\ket{a}$), and Eve's first attack $\mathcal{F}$, and after $B$ measures in the $Z$ basis (recall, we are currently only interested in a key-distillation iteration and so we condition on the event that $B$ chooses $\MR$), the joint state is found to be:
\begin{align*}
&\frac{1}{2}\kb{0}_A\otimes(\kb{0}_B\otimes q_0^2\kb{0,0}_{TE} + \kb{1}_B\otimes q_1^2 \kb{1,e}_{TE})\\
+&\frac{1}{2}\kb{1}_A\otimes(\kb{0}_B\otimes P(q_0\alpha\ket{0,0}_{TE} + q_2\beta\ket{0,f}_{TE})\\
&+ \kb{1}_B\otimes P(q_1\alpha\ket{1,e}_{TE} + q_3\beta\ket{1,0}_{TE})),
\end{align*}
where $P(\ket{z}) = \kb{z} = \ket{z}\bra{z}$.  Following this, the qubit returns to $A$; however, before arriving, Eve has a second opportunity to attack using operator $U_R$.  We write the action of $U_R$ abstractly as:
\begin{align}
U_R\ket{0,0}_{TE} = \ket{0,e_0} + \ket{1,e_1} && U_R\ket{1,0} &= \ket{0,e_2} + \ket{1,e_3}\label{eq:UR-action}\\
U_R\ket{1,e}_{TE} = \ket{0,f_0} + \ket{1,f_1} && U_R\ket{0,f} &= \ket{0,f_2} + \ket{1,f_3}.\notag
\end{align}
Above, the states $\ket{e_i}$ and $\ket{f_i}$ are arbitrary states in $\mathcal{H}_E$ (though, unitarity of $U_R$ imposes some restrictions on them which will be important momentarily).

Following the application of this attack, the qubit returns to $A$ who simply discards it (recall, we are conditioning on an iteration that leads to a raw-key bit).  Thus, we may simply trace out the Transit space following the application of $U_R$.  The final density operator, therefore, is found to be:
\begin{align}
\rho_{ABE} &= \frac{1}{2}\kb{0}_A\otimes(\kb{0}_B\otimes q_0^2(\kb{e_0} + \kb{e_1})\label{eq:density-full}\\
&+ \kb{1}_B\otimes q_1^2(\kb{f_0} + \kb{f_1}))\notag\\
&+\frac{1}{2}\kb{1}_A\otimes(\kb{0}_B\otimes[ P(q_0\alpha\ket{e_0} + q_2\beta\ket{f_2})\notag\\
&+ P(q_0\alpha\ket{e_1} + q_2\beta\ket{f_3})])\notag\\
&+\frac{1}{2}\kb{1}_A\otimes(\kb{1}_B\otimes[ P(q_1\alpha\ket{f_0} + q_3\beta\ket{e_2})\notag\\
&+ P(q_1\alpha\ket{f_1} + q_3\beta\ket{e_3})]).\notag
\end{align}
To clean up the notation, we define the following vectors:
\begin{align*}
\ket{g_0} = q_1\alpha\ket{f_1} + q_3\beta\ket{e_3} && \ket{g_1} = q_1\alpha\ket{f_0} + q_3\beta\ket{e_2}\\
\ket{g_2} = q_0\alpha\ket{e_1} + q_2\beta\ket{f_3} && \ket{g_3} = q_0\alpha\ket{e_0} + q_2\beta\ket{f_2}
\end{align*}
From this, we may then use Theorem \ref{thm:entropy} to derive the following lower-bound:
\begin{align}
 &S(A|E)_\rho \ge \frac{q_0^2\bk{e_0} + \bk{g_0}}{2}\label{eq:SAE-full}\\
&\times\left(H\left[\frac{q_0^2\bk{e_0}}{q_0^2\bk{e_0} + \bk{g_0}}\right] - H\left[\lambda(q_0\ket{e_0},\ket{g_0})\right]\right)\notag\\\notag\\
&+\frac{q_0^2\bk{e_1} + \bk{g_1}}{2}\notag\\
&\times\left(H\left[\frac{q_0^2\bk{e_1}}{q_0^2\bk{e_1} + \bk{g_1}}\right] - H\left[\lambda(q_0\ket{e_1},\ket{g_1})\right]\right)\notag\\\notag\\
&+\frac{q_1^2\bk{f_0} + \bk{g_2}}{2}\notag\\
&\times\left(H\left[\frac{q_1^2\bk{f_0}}{q_1^2\bk{f_0} + \bk{g_2}}\right] - H\left[\lambda(q_1\ket{f_0},\ket{g_2})\right]\right)\notag\\\notag\\
&+\frac{q_1^2\bk{f_1} + \bk{g_3}}{2}\notag\\
&\times\left(H\left[\frac{q_1^2\bk{f_1}}{q_1^2\bk{f_1} + \bk{g_3}}\right] - H\left[\lambda(q_1\ket{f_1},\ket{g_3})\right]\right).\notag
\end{align}
Though, by setting $J= \{0\}$ from the theorem, we also have the following (weaker) lower-bound:
\begin{align}
&S(A|E)_\rho\ge\frac{q_0^2\bk{e_0} + \bk{g_0}}{2}\label{eq:SAE}\\
&\times\left(H\left[\frac{q_0^2\bk{e_0}}{q_0^2\bk{e_0} + \bk{g_0}}\right] - H\left[\lambda(q_0\ket{e_0},\ket{g_0})\right]\right).\notag
\end{align}

It is this lower-bound we will actually consider.
To compute $S(A|E)$ (giving us the key-rate), we need to compute, or bound, the inner-products appearing in the above expression, based only on statistics we may observe.

Note that $q_0$ and $q_1$ are both observable parameters.  Indeed, let $\prf_{0,i}$ be the probability that $B$ measures $\ket{i}$ (for $i \in \{0,1\}$) if $A$ initially sent $\ket{0}$.  This is one of the few statistics $A$ and $B$ actually can estimate and is, in fact, the only observable noise statistic in the forward channel (they cannot measure, for example, $\prf_{1,i}$ when $\alpha > 0$).  It is not difficult to see, from Equation \ref{eq:forward-attack}, that $q_0^2 = \prf_{0,0}$ and $q_1^2 = \prf_{0,1}$.  Note that, by definition of the restricted attack, it is sufficient to consider non-negative $q_i$ \cite{SQKD-entangle}.

As mentioned, the users cannot directly observe $q_2$ and $q_3$.  However, they can estimate it by considering $\prf_{a,1}$, the probability that $B$ measures $\ket{1}$ if $A$ initially sent $\ket{a}$ (this is something that may be observed).  Note that, from Equation \ref{eq:forward-attack-a}, we have:
\begin{align}
\prf_{a,1} &= ||q_1\alpha\ket{e} + q_3\beta\ket{0}||^2\notag\\
&= q_1^2\alpha^2 + q_3^2\beta^2 + 2q_1q_3\alpha\beta Re\braket{0|e}\notag\\
&=\prf_{0,1}\alpha^2 + q_3^2\beta^2 + 2\sqrt{\prf_{0,1}}q_3\alpha\beta Re\braket{0|e}.
\end{align}
Of course, $|\braket{0|e}| \le 1$.  We are constrained by the fact that $q_3 \ge 0$ (since, for the restricted attack, each $q_i$ are non-negative real numbers \cite{SQKD-entangle}).  We therefore have the following solution for $q_3$, looking for the smallest positive root of the above quadratic equation, assuming $\prf_{a,1} \ge \alpha^2\prf_{0,1}$ (which it will be in our evaluations):
\begin{equation}\label{eq:q3}
1\ge q_3 \ge \frac{1}{\beta}\left(\sqrt{\prf_{a,1}} - \alpha\sqrt{\prf_{0,1}}\right).
\end{equation}

We therefore have values, or bounds, for all $q_i$ (note that $q_2 = \sqrt{1-q_3^2}$).  It is clear that we may observe $\bk{e_0}, \bk{e_1}, \bk{f_0}$, and $\bk{f_1}$.  Indeed, let $\prr_{i,j,k}$ denote the probability that $A$'s measurement observes ``$k$'' conditioned on the event $A$ initially sent $\ket{i}$ and $B$ chose $\MR$ and actually observed $\ket{j}$.  Of course, $i \in \{0,a\}$, $j\in\{0,1\}$ and $k\in\{0,a,?\}$.  It is not difficult to see, then, that $\prr_{0,0,0} = p\cdot\bk{e_0}$ where $p$ is the POVM parameter as described in Protocol \ref{prot}; as discussed, we assume $p > 0$.  By unitarity we also have $\bk{e_1} = 1-\bk{e_0}$.  Similarly, we have $\prr_{0,1,0} = p\cdot \bk{f_0}$ and $\bk{f_1} = 1-\bk{f_0}$.  To simplify notation, at this point we will assume a symmetric attack and define the following:
\begin{align*}
&\prr_{0,0,0} = p\cdot(1-Q_R) &&\Rightarrow \bk{e_0} = 1-Q_R\\
& &&\Rightarrow \bk{e_1} = Q_R\\
&\prr_{0,1,0} = p\cdot Q_R &&\Rightarrow \bk{f_0} = Q_R\\
& && \Rightarrow \bk{f_1} = 1-Q_R.
\end{align*}
(Note we use $Q_R$ to denote the noise in the Reverse channel, from $B$ to $A$.)

This assumption that the \emph{observable} noise is symmetric in this manner (which may be enforced by $A$ and $B$ and is a common assumption in (S)QKD security proofs) is not necessary, and our analysis below follows without it; we only use this to simplify notation.  Note that, if there is no noise in the forward channel (in which case $\prr_{0,1,0}$ is technically undefined since we are conditioning on an event which never occurs), then $\bk{f_0}$ and $\bk{f_1}$ never show up in any of our computations and so we may define $\prr_{0,1,0}$ arbitrarily; thus we assume $\prr_{0,1,0} = p\cdot Q_R$ in this case regardless.

We also claim $\bk{g_i}$ may be observed.  Consider the case that $A$ sends $\ket{a}$, $B$ chooses $\MR$ and observes $\ket{0}$.  From Equation \ref{eq:forward-attack-a}, we see the state collapses to:
\[
\frac{\ket{0}(q_0\alpha\ket{0}_E + q_2\beta\ket{f}_E)}{\sqrt{\prf_{a,0}}}.
\]
After Eve attacks the returning qubit, the state is found to be (before $A$ measures):
\[
\frac{\ket{0,g_3} + \ket{1,g_2}}{\sqrt{\prf_{a,0}}}.
\]
Then, it follows that when $A$ measures we have: $\prr_{a,0,0} = p\cdot\bk{g_3}/\prf_{a,0}$. Furthermore, due to unitarity of Eve's attack, it holds that $\bk{g_2} + \bk{g_3} = \prf_{a,0}$.  Repeating the above analysis conditioning on $B$ observing $\ket{1}$, we conclude:
\begin{align}
\bk{g_3} &= \frac{\prf_{a,0}\prr_{a,0,0}}{p} = \prf_{a,0}(1-Q_R)\label{eq:gi}\\
\bk{g_2} &= \prf_{a,0}\left(1-\frac{\prr_{a,0,0}}{p}\right) = \prf_{a,0}Q_R\notag\\
\bk{g_1} &= \frac{\prf_{a,1}\prr_{a,1,0}}{p} = \prf_{a,1}Q_R\notag\\
\bk{g_0} &= \prf_{a,1}\left(1-\frac{\prr_{a,1,0}}{p}\right) = \prf_{a,1}(1-Q_R).\notag
\end{align}
Note that, above, we assumed $\prf_{a,0,0} = p\cdot(1-Q_R)$ and $\prf_{a,1,0} = p\cdot Q_R$.  This symmetry assumption (which also may be enforced by the users) is not necessary but only done to simplify our notation.
Also, as before, if, for instance, $\prf_{a,1} = 0$, then $\ket{g_0}$ and $\ket{g_1}$ technically never appear in $\rho$ and so they may be arbitrary; in this case we may simply define $\prr_{a,1,0} = p\cdot Q_R$.  Similarly for the case if $\prf_{a,0} = 0$.

Finally, to compute our bound on $S(A|E)$, we will also need to compute the inner product appearing in the $\lambda$ function, namely $Re^2\braket{e_0|g_0}$.  As we are interested in the worst-case, we actually want to find a lower-bound on this inner-product (which, as can be seen from Equation \ref{eq:lambda}, minimizes $S(A|E)$).  It is not difficult to see that:
\begin{align}
&Re^2\braket{e_0|g_0} = (q_1\alpha\braket{e_0|f_1} + q_3\beta\braket{e_0|e_3})^2\label{eq:e0g0}\\
&\ge\left[\max\left(0,q_3\beta\braket{e_0|e_3} - q_1\alpha\sqrt{\bk{e_0}\bk{f_1}}\right)\right]^2\notag\\
&\ge\left[\max\left(0, q_3\beta\braket{e_0|e_3} - \alpha\sqrt{\prf_{0,1}}(1-Q_R)\right)\right]^2,\notag
\end{align}
where, above, we used the fact that $|\braket{e_0|f_1}| \le \sqrt{\bk{e_0}\bk{f_1}}$.

We thus reduced the problem to bounding $\braket{e_0|e_3}$.  To attain this, we must look at several more statistics.  First, consider $\bk{g_1} = \prf_{a,1}\prr_{a,1,0}/p$. Expanding $\bk{g_1}$ yields:
\begin{align*}
&\frac{\prf_{a,1}\prr_{a,1,0}}{p} = q_1^2\alpha^2\bk{f_0} + q_3^2\beta^2\bk{e_2}\\
&+ 2q_1q_3\alpha\beta Re\braket{f_0|e_2}.\\\\
\Rightarrow & q_3^2\beta^2\left(\sqrt{\bk{e_2}}\right)^2 + 2q_1q_3\alpha\beta\sqrt{Q_R}\sqrt{\bk{e_2}}\cos\theta\\
&+ q_1^2\alpha^2Q_R - \prf_{a,1}\prr_{a,1,0}/p = 0.
\end{align*}
Above, we used the fact that $Re\braket{f_0|e_2} = \sqrt{\bk{e_2}}\sqrt{\bk{f_0}}\cos\theta = \sqrt{\bk{e_2}}\sqrt{Q_R}\cos\theta$, for some $\theta$ (this follows from the Cauchy-Schwarz inequality).  Solving the above quadratic, and taking the maximal root over all $\theta$ (note that $\bk{e_2}$ represents the probability of a $\ket{1}$ flipping to a $\ket{0}$, however this noise value is not observable and so we can only bound it based on values we can observe), we find:
\begin{align}\label{eq:e2}
\sqrt{\bk{e_2}} &\le \frac{1}{q_3\beta}\left(q_1\alpha\sqrt{Q_R} + \sqrt{\prf_{a,1}\prr_{a,1,0}/p}\right)\\
&=\frac{1}{q_3\beta}\left(q_1\alpha\sqrt{Q_R} + \sqrt{\prf_{a,1}Q_R}\right)\notag
\end{align}
Similarly, we can bound $\bk{f_3}$ by considering $\bk{g_2}=\prf_{a,0}Q_R$. Solving the resulting quadratic, we find:
\begin{equation}\label{eq:f3}
\sqrt{\bk{f_3}} \le \frac{1}{q_2\beta}\left(q_0\alpha\sqrt{Q_R} + \sqrt{\prf_{a,0}Q_R}\right).
\end{equation}
We now have upper-bounds on the ``hidden'' noise of the channel.

Next, let us consider the statistic $\prr_{a,R,a}$ which we use to denote the probability that, conditioning on the event $A$ sends $\ket{a}$, $B$ chooses to $\R$, and $A$ chooses to measure using POVM $\Lambda$ (see Protocol \ref{prot}), that the outcome of this measurement is ``$a$''.

It is straight-forward (though slightly tedious) algebra, to find that:
\begin{equation}
U_R\mathcal{F}\ket{a} = \ket{a}(\ket{V_{a,0,a}} + \ket{V_{a,1,a}}) + \ket{\ba}\ket{E_{\ba}},
\end{equation}
where:
\begin{align}
\ket{V_{a,0,a}} &= q_0\alpha^2\ket{e_0} + q_2\alpha\beta\ket{f_2} + q_0\alpha\beta\ket{e_1} + q_2\beta^2\ket{f_3}\\
\ket{V_{a,1,a}} &= q_1\alpha^2\ket{f_0} + q_3\alpha\beta\ket{e_2} + q_1\alpha\beta\ket{f_1} + q_3\beta^2\ket{e_3},
\end{align}
and where $\ket{E_{\ba}}$ is a sub-normalized vector in $\mathcal{H}_E$, the exact state of which may be found by tracing the action of linear operator $U_R\mathcal{F}$, though its state is irrelevant to our discussion.  From this, we find:
\begin{align}
&\prr_{a,R,a} = p||\ket{V_{a,0,a}} + \ket{V_{a,1,a}}||^2\notag\\
&= p(\bk{V_{a,0,a}} + \bk{V_{a,1,a}} + 2Re\braket{V_{a,0,a}|V_{a,1,a}}).\label{eq:pra-first}
\end{align}

At this point, we must consider additional mismatched measurements.  Consider $\prr_{a,0,a}$ which we use to denote the probability that, conditioning on $A$ sending $\ket{a}$, $B$ choosing $\MR$ and actually observing $\ket{0}$, and $A$ choosing to measure, that she receives outcome ``$a$''.  To compute this probability, we trace the evolution of the qubit as it travels:
\begin{align*}
\ket{a} &\rightarrow \frac{\ket{0}(q_0\alpha\ket{0} + q_2\beta\ket{f})}{\sqrt{\prf_{a,0}}}\\
&\rightarrow \frac{q_0\alpha(\ket{0,e_0}+\ket{1,e_1}) + q_2\beta(\ket{0,f_2} + \ket{1,f_3})}{\sqrt{\prf_{a,0}}}\\
&=\frac{\ket{a}\ket{V_{a,0,a}}}{\sqrt{\prf_{a,0}}} + \ket{\bar{a}}\ket{E'}_E,
\end{align*}
where $\ket{E'}_E$ is some irrelevant, sub-normalized, state in $E$'s ancilla.
Note that, from the above expression, the choice of notation for $\ket{V_{a,0,a}}$ is clear and we find:
\begin{equation}
\prr_{a,0,a} = p\cdot \frac{\bk{V_{a,0,a}}}{\prf_{a,0}} \Rightarrow \bk{V_{a,0,a}} = \frac{\prf_{a,0}\prr_{a,0,a}}{p}.
\end{equation}
Repeating the above but considering the event when $B$ observes $\ket{1}$, we find:
\begin{equation}
\bk{V_{a,1,a}} = \frac{\prf_{a,1}\prr_{a,1,a}}{p}.
\end{equation}
Substituting this into Equation \ref{eq:pra-first} and also expanding $Re\braket{V_{a,0,a}|V_{a,1,a}}$, we find:
\begin{align}
\prr_{a,R,a} &= \prf_{a,0}\prr_{a,0,a} + \prf_{a,1}\prr_{a,1,a}\label{eq:QA-noise-full}\\
&+2p\cdot Re(q_0q_1\alpha^4\braket{e_0|f_0} + q_0q_3\alpha^3\beta\braket{e_0|e_2})\notag\\
&+2p\cdot Re( q_0q_1\alpha^3\beta\braket{e_0|f_1} + q_0q_3\alpha^2\beta^2\braket{e_0|e_3})\notag\\
&+2p\cdot Re(q_1q_2\alpha^3\beta\braket{f_0|f_2} + q_2q_3\alpha^2\beta^2\braket{e_2|f_2})\notag\\
&+2p\cdot Re( q_1q_2\alpha^2\beta^2\braket{f_1|f_2} + q_2q_3\alpha\beta^3\braket{e_3|f_2})\notag\\
&+2p\cdot Re(q_0q_1\alpha^3\beta\braket{e_1|f_0} + q_0q_3\alpha^2\beta^2\braket{e_1|e_2})\notag\\
&+ 2p\cdot Re(q_0q_1\alpha^2\beta^2\braket{e_1|f_1} + q_0q_3\alpha\beta^3\braket{e_1|e_3})\notag\\
&+2p\cdot Re(q_1q_2\alpha^2\beta^2\braket{f_0|f_3} + q_2q_3\alpha\beta^3\braket{e_2|f_3})\notag\\
&+2p\cdot Re( q_1q_2\alpha\beta^3\braket{f_1|f_3} + q_2q_3\beta^4\braket{e_3|f_3}).\notag
\end{align}
(Note that, above, we used the fact that $Re\braket{x|y} = Re\braket{y|x}$.)  We may simplify the above equation slightly by taking advantage of the unitarity of $U_R$.  Namely, we have the following restrictions (see Equation \ref{eq:UR-action}):
\begin{align*}
\braket{e_0|e_2} + \braket{e_1|e_3} = 0 && &\braket{f_0|f_2} + \braket{f_1|f_3} = 0\\
\braket{e_0|f_0} + \braket{e_1|f_1} = 0 && &\braket{e_2|f_2} + \braket{e_3|f_3} = 0.
\end{align*}
Using this, Equation \ref{eq:QA-noise-full} becomes:
\begin{align}
\prr_{a,R,a} &= \prf_{a,0}\prr_{a,0,a} + \prf_{a,1}\prr_{a,1,a}\\
&+2p\cdot Re(q_0q_1[\alpha^4 - \alpha^2\beta^2]\braket{e_0|f_0})\notag\\
& +2p\cdot Re( q_0q_3[\alpha^3\beta - \alpha\beta^3]\braket{e_0|e_2})\notag\\
&+2p\cdot Re(q_1q_2[\alpha^3\beta-\alpha\beta^3]\braket{f_0|f_2})\notag\\
&+2p\cdot Re( q_2q_3[\alpha^2\beta^2 - \beta^4]\braket{e_2|f_2})\notag\\
&+2p\cdot Re(q_0q_3\alpha^2\beta^2[\braket{e_0|e_3} + \braket{e_1|e_2}])\notag\\
&+2p\cdot Re( q_0q_1\alpha^3\beta[\braket{e_0|f_1} + \braket{e_1|f_0}])\notag\\
&+2p\cdot Re(q_1q_2\alpha^2\beta^2[\braket{f_0|f_3} + \braket{f_1|f_2}])\notag\\
& +2p\cdot Re( q_2q_3\alpha\beta^3[\braket{e_2|f_3} + \braket{e_3|f_2}]).\notag
\end{align}
Consider the following inner-product:
\begin{align*}
&Re\braket{g_1|g_3} = q_0q_1\alpha^2Re\braket{e_0|f_0} + q_1q_2\alpha\beta Re\braket{f_0|f_2}\\
&+ q_0q_3\alpha\beta Re\braket{e_0|e_2} + q_2q_3\beta^2Re\braket{e_2|f_2}.
\end{align*}
Then the above equation for $\prr_{a,R,a}$ simplifies to:
\begin{align}
\prr_{a,R,a} &= \prf_{a,0}\prr_{a,0,a} + \prf_{a,1}\prr_{a,1,a}\\
&+ 2p(\alpha^2-\beta^2)Re\braket{g_1|g_3}\notag\\
&+2p\cdot Re(q_0q_3\alpha^2\beta^2[\braket{e_0|e_3} + \braket{e_1|e_2}])\notag\\
&+2p\cdot Re( q_0q_1\alpha^3\beta[\braket{e_0|f_1} + \braket{e_1|f_0}])\notag\\
&+2p\cdot Re(q_1q_2\alpha^2\beta^2[\braket{f_0|f_3} + \braket{f_1|f_2}])\notag\\
&+2p\cdot Re( q_2q_3\alpha\beta^3[\braket{e_2|f_3} + \braket{e_3|f_2}]).\notag
\end{align}
Solving for the term involving $Re\braket{e_0|e_3}$ (which is the quantity we are currently interested in bounding) yields:
\begin{align}
&q_0q_3\alpha^2\beta^2Re\braket{e_0|e_3}\label{eq:e0e3}\\
&=\frac{1}{2p}(\prr_{a,R,a} - \prf_{a,0}\prr_{a,0,a} - \prf_{a,1}\prr_{a,1,a})\notag\\
&- (\alpha^2-\beta^2)Re\braket{g_1|g_3} - \chi\notag
\end{align}
where:
\begin{align}
\chi &= Re(q_0q_1\alpha^3\beta[\braket{e_0|f_1} + \braket{e_1|f_0}])\notag\\
&+ Re(q_0q_3\alpha^2\beta^2\braket{e_1|e_2})\notag\\
&+ Re(q_1q_2\alpha^2\beta^2[\braket{f_0|f_3} + \braket{f_1|f_2}])\notag\\
&+ Re(q_2q_3\alpha\beta^3[\braket{e_2|f_3} + \braket{e_3|f_2}]).\notag
\end{align}
$A$ and $B$ do not have sufficient quantum capabilities to fully bound $\chi$; however we can bound it based on what we already know and using the Cauchy-Schwarz inequality, namely:
\begin{align}
|\chi| &\le q_0q_1\alpha^3\beta[(1-Q_R) + Q_R]\label{eq:chi}\\
&+ q_0q_3\alpha^2\beta^2\sqrt{Q_R\bk{e_2}}\notag\\
&+ q_1q_2\alpha^2\beta^2\sqrt{Q_R\bk{f_3}}\notag\\
&+ q_1q_2\alpha^2\beta^2\sqrt{(1-Q_R)(1-\bk{f_3})}\notag\\
&+ q_2q_3\alpha\beta^3\sqrt{\bk{e_2}\bk{f_3}}\notag\\
&+ q_2q_3\alpha\beta^3\sqrt{(1-\bk{e_2})(1-\bk{f_3})}.\notag
\end{align}
(Note that, above, we used the fact that $\bk{e_3} = 1-\bk{e_2}$ and $\bk{f_2} = 1-\bk{f_3}$.)  Upper-bounds on $\bk{e_2}$ and $\bk{f_3}$ were already derived in Equations \ref{eq:e2} and \ref{eq:f3}.

Finally, we claim $A$ and $B$ can observe $Re\braket{g_1|g_3}$ by considering the statistic $\prr_{a,R,0}$; that is, the probability that $A$'s measurement produces outcome ``$0$'' conditioned on the event she initially sent $\ket{a}$ and $B$ chose $\R$.  Indeed, tracing the qubit in this case, it is not difficult to see that:
\[
U_R\mathcal{F}\ket{a} = \ket{0}(\ket{g_1} + \ket{g_3}) + \ket{1}(\ket{g_0}+\ket{g_2}),
\]
from which we attain:
\begin{align}
\prr_{a,R,0} &= p||\ket{g_1} + \ket{g_3}||^2\notag\\
&= p(\bk{g_1} + \bk{g_3} + 2Re\braket{g_1|g_3})\notag\\
\Rightarrow  Re\braket{g_1|g_3} &= \frac{1}{2}\left(\frac{\prr_{a,R,0}}{p} - \bk{g_1} - \bk{g_3}\right).
\end{align}
Since $\bk{g_i}$ are all observable (see Equation \ref{eq:gi}), this completes our bound.

This completes our lower-bound on $S(A|E)$.  To summarize, given as input $\alpha$ along with those observable statistics as utilized above, one must simply minimize Equation \ref{eq:SAE} over all $q_3$, $\bk{e_2}$, and $\bk{f_3}$, as enforced by Equations \ref{eq:q3}, \ref{eq:e2}, and \ref{eq:f3}.  For any particular choice of these values, one may compute a bound on $\chi$ from Equation \ref{eq:chi}; one may also compute a bound on $Re\braket{e_0|e_3}$ using Equation \ref{eq:e0e3}.  This then allows one to bound $Re^2\braket{e_0|g_0}$, using Equation \ref{eq:e0g0} which gives a possible value of $S(A|E)$.  Minimizing over $\bk{e_2}$, $\bk{f_3}$, and $q_3$ gives a worst-case lower-bound on $S(A|E)$ over all attacks which induce the observed statistics.  This is a simple minimization problem allowing one to evaluate the key-rate numerically.

Note that if $\alpha = 0$ (i.e., the protocol is classical), then it is easy to check that Equation \ref{eq:e0e3} becomes simply $0 = 0$, regardless of the choice of $\braket{e_0|e_3}$ (i.e., Eve may set this inner-product arbitrarily and Equation \ref{eq:e0e3} will be satisfied).  It is also clear that $\braket{e_0|g_0} = q_3\braket{e_0|e_3}$.  Thus, Eve may set $\braket{e_0|e_3} = 0$ in this case resulting in the entropy $S(A|E) = 0$ as expected.  That is, in the classical case, Eve has no uncertainty on $A$ and $B$'s raw key and so the protocol is insecure.  The interesting question is what happens when $\alpha > 0$?

To finish the key-rate computation (and answer this question), we also need $H(A|B)$, however this value is easily found:
\begin{align}
&H(A|B)\\
&= H\left(\frac{q_0^2}{2}, \frac{q_1^2}{2}, \frac{\bk{g_0}+\bk{g_1}}{2}, \frac{\bk{g_2}+\bk{g_3}}{2}\right)\notag\\
&- H\left(\frac{q_0^2 + \bk{g_2}+\bk{g_3}}{2}\right)\notag\\\notag\\
&=H\left(\frac{\prf_{0,0}}{2}, \frac{\prf_{0,1}}{2}, \frac{\prf_{a,1}}{2}, \frac{\prf_{a,0}}{2}\right)\notag\\
&- H\left(\frac{\prf_{0,0} + \prf_{a,0}}{2}\right)\notag
\end{align}
thus completing the key-rate computation.

\section{Evaluation}
To evaluate our protocol, and more importantly to see the effect of $\alpha$ on the secure key-rate, we must put values to those observable statistics $\prf_{\cdot,\cdot}$ and $\prr_{\cdot,\cdot,\cdot}$.  We will assume a symmetric attack parameterized by noise values $Q_F$ (in the forward channel), $Q_R$ (in the reverse), and $Q_X$ (for the ``loop'' channel when $B$ reflects), where:
\begin{align*}
\prf_{0,0} = 1-Q_F && &\prf_{0,1} = Q_F\\
\prr_{0,0,0}/p = 1-Q_R && &\prr_{0,1,0}/p = Q_R\\
\prr_{a,0,0}/p = 1-Q_R && &\prr_{a,1,0}/p = Q_R\\
\prr_{a,R,a}/p = 1-Q_X.
\end{align*}
and where $p = \frac{1}{1+\alpha}$, the maximal allowed value as discussed earlier.

To put values to the mismatched events, we model the channel as a depolarization channel, a common approach when evaluating (S)QKD protocols.  This is not a requirement of our security proof of course, simply a way to put realistic (i.e., physically realizable) numbers to the observable parameters in order to evaluate the key-rate.  A depolarization channel with parameter $Q$ is simply the map:
\[
\mathcal{E}_Q(\rho) = (1-2Q)\rho + Q\cdot I.
\]
From this, we find:
\begin{align*}
\prf_{a,0} &= (1-2Q_F)\alpha^2 + Q_F\\
\prf_{a,1} &= (1-2Q_F)\beta^2 + Q_F\\
\prr_{a,0,a}/p &= (1-2Q_R)\alpha^2 + Q_R\\
\prr_{a,1,a}/p &= (1-2Q_R)\beta^2 + Q_R\\
\prr_{a,R,0}/p &= (1-2Q_Z)\alpha^2 + Q_Z.
\end{align*}

As expected, the noise tolerance of this protocol is low, however we are able to attain positive key-rates as shown in Figures \ref{fig:graph3}, \ref{fig:graph1}, and \ref{fig:graph2}.  It is clear from these figures that the forward channel noise is the most important statistic - indeed, as shown in Figure \ref{fig:graph3}, the protocol can tolerate a high level of reverse and ``loop'' noise (approaching $10\%$).  However, as shown in Figure \ref{fig:graph2}, if the forward channel increases too much (even by a small amount), there are only a few choices for $\alpha$ where a positive key-rate can be attained (and that key-rate is still low).  Unless the reverse channel noise is very large, the optimal choice for $\alpha$ ranged between $0.175$ and $0.2$ for those evaluations we performed.  For small $Q_F$ and high $Q_R$ and $Q_X$, as in Figure \ref{fig:graph3}, the optimal value of $\alpha$ is slightly lower, ranging between $0.13$ and $0.16$.

Despite the low noise tolerance, we still consider this a positive, and interesting, result as this protocol was designed specifically to smoothly transform from classical to quantum communication and to allow research in investigating how this affects secure communication.  Of course, our key-rate is a lower bound, so the actual security rate can only be higher.  Further studying this would make interesting future work.

\begin{figure}
  \centering
  \includegraphics[width=200pt]{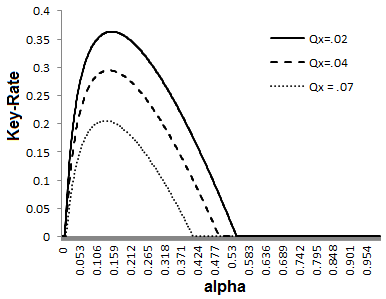}
  \caption{Key-rate when the forward channel noise is close to zero ($10^{-5}$) and the reverse and loop noise levels are high.  We see that the forward channel noise is the most critical for this protocol.}\label{fig:graph3}
\end{figure}

\begin{figure}
  \centering
  \includegraphics[width=200pt]{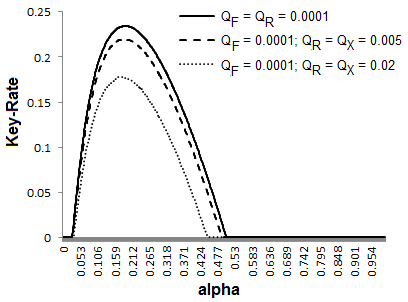}
  \caption{Key-rate for low forward channel noise (though higher than Figure \ref{fig:graph3}) and increasing reverse and loop noise.}\label{fig:graph1}

\end{figure}
\begin{figure}
  \centering
  \includegraphics[width=200pt]{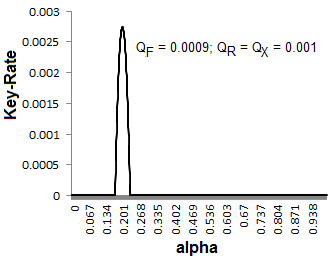}
  \caption{Key-rate when the forward channel noise is increased - only a small window of $\alpha$ values exist in this case when the protocol attains a positive key-rate.}\label{fig:graph2}
\end{figure}


\section{Comments on Further Restrictions}

One natural question for future work is: can the requirements of this protocol be reduced even further?  That is, can $A$ have even more restrictions placed on her quantum abilities?  One clear direction is to attempt to remove $A$'s POVM and replace it with a single basis measurement, measuring in the $\mathcal{A} = \{\ket{a}, \ket{\bar{a}}\}$ basis (where $\braket{\bar{a}|a} = 0$).  However, this removes certain key statistics that we relied on in our security proof.  While we attempted to analyze this protocol, a full security proof remains elusive.

We do, however, conjecture that this even more restricted protocol is secure. To provide at least some evidence in support of this, we were able to analyze a particular intercept-resend attack and show that the protocol is secure against this.  The attack we consider is one which induces no additional noise in the channel (that is, it is undetectable).  To remain hidden from $A$ and $B$, Eve simply measures the reverse channel in the $\mathcal{A}$ basis (the same basis $A$ uses, thus $E$ will have the same information as $A$ does from the reverse channel - but, importantly, not the forward channel).  If this measurement results in outcome $\ket{a}$, $E$ guesses the raw key bit is $1$; otherwise she guesses it is $0$ (note that if $\alpha = 0$ this guess is always correct and so, in that case, the protocol is insecure as expected).  We compute the values $\pe_{i,j}$ for $i,j\in\{0,1\}$ which we use to denote the probability that $A$'s raw key bit is $i$ and $E$'s guess is $j$ assuming she uses this attack.  From this we can compute the key-rate equation for any $\alpha$.

The attack schematic is shown in Figure \ref{fig:ir-atk}; the key-rate for various $\alpha$ is shown in Figure \ref{fig:ir-atk-keyrate}.  We notice that the key-rate is positive for all $\alpha \in (0,1)$ (of course it is insecure if $\alpha=0$ or $1$).  The optimal choice for $\alpha$ in this event is $\alpha = 0.5$ (contrast this with the ``full'' protocol we analyzed in this paper where the optimal was usually around $0.2$).  Also note the asymmetry in the key-rate graph.

Of course this is only showing some evidence that this further restriction (i.e., removing $A$'s ability to use POVM $\Lambda$ as we considered in our protocol in this work) may result in a secure protocol.  A complete analysis we leave as interesting future work.

\begin{figure}
  \centering
  \includegraphics[width=200pt]{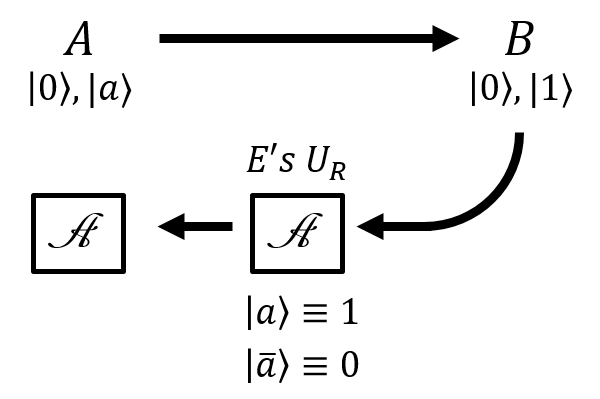}
  \caption{An intercept-resend attack against an even more restricted protocol than the one we analyzed here.  This attack induces no observable noise in the quantum channel; Eve simply measures in the same basis that $A$ will.  If she observes $\ket{a}$, she will guess that the raw key is $1$; otherwise she guesses it is $0$ (for $\alpha=0$ this gives Eve full information).}\label{fig:ir-atk}
\end{figure}

\begin{figure}
  \centering
  \includegraphics[width=250pt]{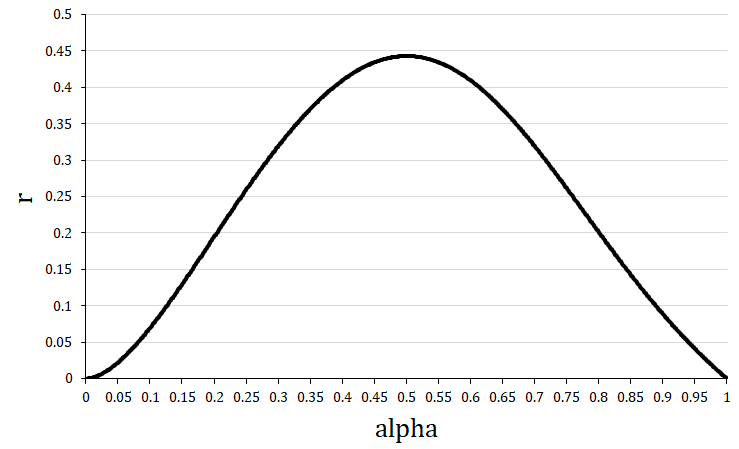}
  \caption{Showing how the parameter $\alpha$ affects this further restricted protocol's key-rate $r$ (where POVM $\Lambda$ is replaced with only a single basis measurement) when $E$ uses the intercept-resend attack discussed in the text.  Note the slight asymmetry in the graph.}\label{fig:ir-atk-keyrate}
\end{figure}

\section{Closing Remarks}

In this paper, we developed a new SQKD protocol with a tuneable parameter $\alpha$ allowing one to gauge the effect of the secure communication rate, based on ``how quantum'' the protocol is.  When $\alpha$ is set to zero, the communication is purely classical and thus the protocol is insecure.  As $\alpha$ increases, security can be attained for certain optimal choices and for certain channels.  Studying the protocol further may help to shed light on the ``gap'' between quantum and classical secure communication.  Furthermore, our proof approach may be applicable to other (S)QKD protocols where users are highly restricted in their quantum capabilities (either intentionally or due, perhaps, to hardware faults).

Many interesting future problems remain open.  Obviously the noise tolerance of our protocol is very low - though, we stress that we are only interested in this protocol from a theoretical perspective and in discovering when, or even if, this protocol can be secure (and our answer is in the affirmative).  However, it would be interesting to try to improve on this. Our bound may be improved by attempting to bound all terms appearing in Equation \ref{eq:SAE-full} (we only used the lower bound from Equation \ref{eq:SAE}).  Other mismatched statistics may help here.  Also, studying the effect of $\alpha$ against different forms of attacks (e.g., practical intercept-resend attacks) may also be highly beneficial and interesting.

Another interesting question is whether we can reduce the resource requirements of the users even further.  As commented on in the previous section, we attempted to analyze the case where $A$'s measurement capabilities are further reduced than what we used in this paper; so far, however, a full proof of security in that case remains an open problem.

\textbf{Acknowledgments:} AG would like to acknowledge the support of National Science Foundation grant number 1659764, which supported her during a summer REU at the University of Connecticut.  WK and HI are partially supported by the NSF under grant number 1812070.



\begin{thebibliography}{10}
\providecommand{\url}[1]{#1}
\csname url@samestyle\endcsname
\providecommand{\newblock}{\relax}
\providecommand{\bibinfo}[2]{#2}
\providecommand{\BIBentrySTDinterwordspacing}{\spaceskip=0pt\relax}
\providecommand{\BIBentryALTinterwordstretchfactor}{4}
\providecommand{\BIBentryALTinterwordspacing}{\spaceskip=\fontdimen2\font plus
\BIBentryALTinterwordstretchfactor\fontdimen3\font minus
  \fontdimen4\font\relax}
\providecommand{\BIBforeignlanguage}[2]{{%
\expandafter\ifx\csname l@#1\endcsname\relax
\typeout{** WARNING: IEEEtran.bst: No hyphenation pattern has been}%
\typeout{** loaded for the language `#1'. Using the pattern for}%
\typeout{** the default language instead.}%
\else
\language=\csname l@#1\endcsname
\fi
#2}}
\providecommand{\BIBdecl}{\relax}
\BIBdecl

\bibitem{SQKD-first-PRL}
M.~Boyer, D.~Kenigsberg, and T.~Mor, ``Quantum key distribution with classical
  bob,'' \emph{Phys. Rev. Lett.}, vol.~99, p. 140501, Oct 2007.

\bibitem{QKD-Tom-Krawec-Arbitrary}
W.~O. Krawec, ``Quantum key distribution with mismatched measurements over
  arbitrary channels,'' \emph{Quantum Information and Computation}, vol.~17,
  no. 3 and 4, pp. 209--241, 2017.

\bibitem{QKD-survey}
V.~Scarani, H.~Bechmann-Pasquinucci, N.~J. Cerf, M.~Du\ifmmode~\check{s}\else
  \v{s}\fi{}ek, N.~L\"utkenhaus, and M.~Peev, ``The security of practical
  quantum key distribution,'' \emph{Rev. Mod. Phys.}, vol.~81, pp. 1301--1350,
  Sep 2009.

\bibitem{QKD-renner-keyrate}
\BIBentryALTinterwordspacing
R.~Renner, N.~Gisin, and B.~Kraus, ``Information-theoretic security proof for
  quantum-key-distribution protocols,'' \emph{Phys. Rev. A}, vol.~72, p.
  012332, Jul 2005. [Online]. Available:
  \url{http://link.aps.org/doi/10.1103/PhysRevA.72.012332}
\BIBentrySTDinterwordspacing

\bibitem{QKD-general-attack}
M.~Christandl, R.~Konig, and R.~Renner, ``Postselection technique for quantum
  channels with applications to quantum cryptography,'' \emph{Phys. Rev.
  Lett.}, vol. 102, p. 020504, Jan 2009.

\bibitem{QKD-general-attack2}
R.~Renner, ``Symmetry of large physical systems implies independence of
  subsystems,'' \emph{Nature Physics}, vol.~3, no.~9, pp. 645--649, 2007.

\bibitem{QKD-winter-keyrate}
I.~Devetak and A.~Winter, ``Distillation of secret key and entanglement from
  quantum states,'' \emph{Proc. of the Royal Society A: Math., Physical and
  Engineering Science}, vol. 461, no. 2053, pp. 207--235, 2005.

\bibitem{QKD-Tom-First}
S.~M. Barnett, B.~Huttner, and S.~J. Phoenix, ``Eavesdropping strategies and
  rejected-data protocols in quantum cryptography,'' \emph{Journal of Modern
  Optics}, vol.~40, no.~12, pp. 2501--2513, 1993.

\bibitem{QKD-Tom-KeyRateIncrease}
S.~Watanabe, R.~Matsumoto, and T.~Uyematsu, ``Tomography increases key rates of
  quantum-key-distribution protocols,'' \emph{Physical Review A}, vol.~78,
  no.~4, p. 042316, 2008.

\bibitem{SQKD-second}
M.~Boyer, R.~Gelles, D.~Kenigsberg, and T.~Mor, ``Semiquantum key
  distribution,'' \emph{Phys. Rev. A}, vol.~79, p. 032341, Mar 2009.

\bibitem{SQKD-lessthan4}
X.~Zou, D.~Qiu, L.~Li, L.~Wu, and L.~Li, ``Semiquantum-key distribution using
  less than four quantum states,'' \emph{Phys. Rev. A}, vol.~79, p. 052312, May
  2009.

\bibitem{SQKD-Single-Security}
W.~O. Krawec, ``Restricted attacks on semi-quantum key distribution
  protocols,'' \emph{Quantum Information Processing}, vol.~13, no.~11, pp.
  2417--2436, 2014.

\bibitem{SQKD-3}
W.~Jian, Z.~Sheng, Z.~Quan, and T.~Chao-Jing, ``Semiquantum key distribution
  using entangled states,'' \emph{Chinese Physics Letters}, vol.~28, no.~10, p.
  100301, 2011.

\bibitem{SQKD-cl-A}
H.~Lu and Q.-Y. Cai, ``Quantum key distribution with classical alice,''
  \emph{International Journal of Quantum Information}, vol.~6, no.~06, pp.
  1195--1202, 2008.

\bibitem{SQKD-no-measure}
X.~Zou, D.~Qiu, S.~Zhang, and P.~Mateus, ``Semiquantum key distribution without
  invoking the classical party’s measurement capability,'' \emph{Quantum
  Information Processing}, vol.~14, no.~8, pp. 2981--2996, 2015.

\bibitem{SQKD-mirror}
\BIBentryALTinterwordspacing
M.~Boyer, M.~Katz, R.~Liss, and T.~Mor, ``Experimentally feasible protocol for
  semiquantum key distribution,'' \emph{Phys. Rev. A}, vol.~96, p. 062335, Dec
  2017. [Online]. Available:
  \url{https://link.aps.org/doi/10.1103/PhysRevA.96.062335}
\BIBentrySTDinterwordspacing

\bibitem{SQKD-Krawec-SecurityProof}
W.~O. Krawec, ``Security proof of a semi-quantum key distribution protocol,''
  in \emph{Information Theory (ISIT), 2015 IEEE International Symposium
  on}.\hskip 1em plus 0.5em minus 0.4em\relax IEEE, 2015, pp. 686--690.

\bibitem{SQKD-zhang2016single}
W.~Zhang, D.~Qiu, X.~Zou, and P.~Mateus, ``A single-state semi-quantum key
  distribution protocol and its security proof,'' \emph{arXiv preprint
  arXiv:1612.03087}, 2016.

\bibitem{QKD-BB84-Modification}
H.-K. Lo, H.-F. Chau, and M.~Ardehali, ``Efficient quantum key distribution
  scheme and a proof of its unconditional security,'' \emph{Journal of
  Cryptology}, vol.~18, no.~2, pp. 133--165, 2005.

\bibitem{SQKD-photon-tag}
\BIBentryALTinterwordspacing
Y.-g. Tan, H.~Lu, and Q.-y. Cai, ``Comment on Òquantum key distribution with
  classical bobÓ,'' \emph{Phys. Rev. Lett.}, vol. 102, p. 098901, Mar 2009.
  [Online]. Available:
  \url{http://link.aps.org/doi/10.1103/PhysRevLett.102.098901}
\BIBentrySTDinterwordspacing

\bibitem{SQKD-photon-tag-comment}
\BIBentryALTinterwordspacing
M.~Boyer, D.~Kenigsberg, and T.~Mor, ``Boyer, kenigsberg, and mor reply:,''
  \emph{Phys. Rev. Lett.}, vol. 102, p. 098902, Mar 2009. [Online]. Available:
  \url{http://link.aps.org/doi/10.1103/PhysRevLett.102.098902}
\BIBentrySTDinterwordspacing

\bibitem{SQKD-entangle}
W.~O. Krawec, ``Key-rate bound of a semi-quantum protocol using an entropic
  uncertainty relation,'' in \emph{2018 IEEE International Symposium on
  Information Theory (ISIT)}, June 2018, pp. 2669--2673.

\end{thebibliography}

\end{document}